\documentclass[onecolumn,10pt,journal,letterpaper,twoside]{IEEEtran}

\usepackage{times}
\usepackage{hyperref} 
\usepackage{amsmath}
\usepackage{amssymb}
\usepackage{amsthm}
\usepackage{bbm}
\usepackage{dsfont}

\usepackage{epstopdf} 
\usepackage{graphicx}
\usepackage{url}
\usepackage{mathrsfs}

\usepackage{color}

\DeclareMathOperator*{\argmin}{arg\,min}

\newtheorem{theorem}{Theorem}
\newtheorem{lemma}{Lemma}
\newtheorem{proposition}[theorem]{Proposition}
\newtheorem{corollary}[theorem]{Corollary}

\theoremstyle{definition}
\newtheorem{definition}{Definition}
\theoremstyle{remark}

\newcommand{\mynewline}{\mbox{}\\}
\newcommand{\E}[1]{\mathbb{E} \left( #1 \right)}

\newcommand{\Es}[2]{\mathbb{E}_{#1} \left( #2 \right)}

\newcommand{\BRA}[1]{\left( #1 \right)}

\newcommand{\abs}[1]{\left| #1 \right|}

\newcommand{\BRAs}[1]{\left\{ #1 \right \}}
\newcommand{\PR}[1]{\mathbb{P}\left\{ #1 \right\}}
\newcommand{\PRs}[2]{\mathbb{P}_{#1}\left\{ #2 \right\}}

\newcommand{\cA}{{\mathcal A}}

\newcommand{\cR}{{\mathcal R}}

\newcommand{\cX}{{\mathcal X}}

\newcommand{\cY}{{\mathcal Y}}

\newcommand{\cC}{{\mathcal C}}

\newcommand{\cU}{{\mathcal U}}

\newcommand{\scP}{\mathscr{P}}

\newcommand{\ie}{{\emph{i.e.}}}
\newcommand{\eg}{{\emph{e.g.}}}

\newcommand{\Ind}[1]{ \mathds{1}_{\BRAs{#1}} }

\newcommand{\MIN}[1]{ \smash{\displaystyle\min_{#1}} }
\newcommand{\SUP}[1]{ \smash{\displaystyle\sup_{#1}} }
\newcommand{\INF}[1]{ \smash{\displaystyle\inf_{#1}} }

\begin{document}

\title{One shot approach to lossy source coding under average distortion constraints}

\author{Nir~Elkayam ~~~~~~~Meir~Feder \\
	Department of Electrical Engineering - Systems\\
	Tel-Aviv University, Israel \\
	Email: nirelkayam@post.tau.ac.il, meir@eng.tau.ac.il}

\maketitle

\newif\ifFullProofs
\FullProofstrue 

\pagenumbering{gobble}

\subsection*{\centering Abstract}
\textit{ This paper presents a one shot analysis of the lossy compression problem under average distortion constraints. We calculate the exact expected distortion of a random code. The result is given as an integral formula using a newly defined functional $\tilde{D}(z,Q_Y)$ where $Q_Y$ is the random coding distribution and $z\in [0,1]$. When we plug in the code distribution as $Q_Y$, this functional produces the average distortion of the code, thus provide a converse result utilizing the same functional. Two alternative formulas are provided for $\tilde{D}(z,Q_Y)$, the first involves a supremum over some auxiliary distribution $Q_X$ which has resemblance to the channel coding meta-converse and the other involves an infimum over channels which resemble the well known Shannon distortion-rate function. }

\section{Introduction}
The single shot approach aims to find informational quantities that govern the optimal performance of an operational problems of interest, \eg, channel coding and lossy compression. In both cases, the problem settings pose a random object that we want to control. In the channel coding problem this is the random channel which abstracts the medium we want to use to enable a reliable communication. In the lossy compression this is the source we want to ``compress'' to a minimum number of bits subject to a distortion constraint. The single shot approach tries to solve the problem by providing achievable and converse bounds without any assumption on the random object. 

A ``good'' solution should have the following properties:
\begin{enumerate}
	\item \textbf{Tightness}: The relation between the achievable and converse bounds should be clarifies and quantified. Preferably, the gap between the bounds should be ``small''. 
	\item \textbf{Computation}: The bounds should be computable. Since we generally deal with a high dimensional problem space for which the exact description might not even be feasible, we relax the computability to convexity, \ie, the bounds should be presented as a minimization of some convex function on some convex domain. For such a problems, symmetries might solve the problem entirely or substantially reduce the effective size, see \eg \cite[Theorem 20]{polyanskiy2013saddle}. 
	\item \textbf{Generalization}: The bounds can be relaxed to other known bounds.  
\end{enumerate} 
In this paper we deal with the lossy source coding problem. In \cite{ElkayamITW2015} we presented a general approach to the one-shot coding problem. In this paper we borrow and extends ideas from \cite{ElkayamITW2015} and provide a novel analysis of the lossy compression problem. We derive an achievable bound using random coding and a corresponding converse bound. Both bounds are given in term of a newly defined functional $\tilde{D}(z,Q_Y)$ where $z \in [0,1]$ and $Q_Y$ denote a distribution over the reproduction space $\cY$. The functional $\tilde{D}(z,Q_Y)$ is shown to be convex and has similarity to the channel coding meta-converse \cite[Theorem 27]{polyanskiy2010channel}. 

\subsection{Notation} \label{sec:nota}

Throughout this paper, scalar random variables are denoted by capital letters (e.g. $X$), sample values are denoted by lower case letters (e.g. $x$) and their alphabets are denoted by their respective calligraphic letters, (e.g. $\cX$). 

The set of all distributions (probability mass  functions) supported on alphabet $\cY$ is denoted as $\scP(\cY)$. The set of all conditional distributions (i.e., channels) with the input alphabet $\cX$ and the output alphabet $\cY$ is denoted by $\scP(\cY|\cX )$. If $X$ has distribution $P_X$, we write this as $X\sim P_X$. The uniform probability distribution over $[0, 1]$ is denoted throughout by $\cU$. The probability (expectation) of an event (random variable) $\cA$ under the distribution $P_X$ is denoted by $\PRs{P_X}{\cA}$ ($\Es{P_X}{\cA}$) respectively, \eg\ $\PRs{P_X}{X \geq \alpha}$ and $\Es{P_X}{f(X)}$. In some cases, we abbreviate the notation and write $P_X\BRAs{\cA}$ instead of $\PRs{P_X}{\cA}$, \eg\ $P_X\BRA{X \geq \alpha}=\PRs{P_X}{X \geq \alpha}$. In some cases, we write $\Es{\mu}{f(X)}$ where $\mu$ is $\sigma-$finite measure and not a probability measure.

\section{Problem setting}
Let $X$ denote the random variable on $\cX$, representing the source we want to compress. The elements of $\cX$ are the input symbols. Denote by $P_X$ the distribution of $X$. Let $\cY$ denote the set of reproduction symbols. Let $d:\cX\times\cY\rightarrow \mathbb{R}^+$ denote the distortion function. 
Any subset $\cC \subset \cY$ is a \textbf{code} and the \textit{average distortion} associated with this code is:
\begin{equation}\label{rate_distortion:average_distortion_definition}
D(P_X, \cC) \triangleq \Es{P_X}{\MIN{y \in \cC}\BRAs{d(X,y)}}
\end{equation}
Let:
\begin{equation}\label{rate_distortion:optimal_average_distortion_definition}
D(P_X, R) \triangleq \MIN{\cC \subset \cY: \abs{\cC}=e^R}D(P_X,\cC) 
\end{equation}
denote the optimal distortion-rate function. The goal is to find upper and lower bounds on $D(P_X, R)$. 

Throughout this paper we assume that both $\cX$ and $\cY$ are finite sets. Thus the distribution $P_X$ is discrete and the distortion is bounded by some $d_{max}$. The results in this paper can be extended quite straight forwardly to rather general alphabets $\cX$, $\cY$ and appropriate $\sigma$-algebras, as long as the probability distribution $P_X$ is well defined. The boundedness of the distortion can be relaxed to the following: There exist a ``small'' finite set $\cR \subset \cY$ such that $\Es{P_X}{\MIN{y\in\cR}d(X,y)} = d_{max} < \infty$.

\section{Achievability bound}

For the achievable argument we use the random coding approach. Let $Q_Y\in\scP(\cY)$ denote a given distribution on $\cY$. A random code of rate $R$ with $M=e^R+1$ codewords is $C=\BRAs{Y_0, \dots, Y_{e^R}}$ where each $Y_i$ is drawn from $Q_Y$ independently of the other code words. The average distortion of the random code is:
\begin{equation}
D(P_X, Q_Y, R)\triangleq \Es{C}{D(P_X, C)}
\end{equation}

\subsection{pairwise correct probability}
In \cite{ElkayamITW2015} the pairwise error probability was defined as the random variable representing the probability of error given the sent and received symbols. Here, we define the \textit{pairwise correct probability}, which represent the probability of drawing a reproduction symbol that is better than a given reproduction symbol.

\begin{definition}
	For $x\in \cX$, $y\in\cY$ and $u\in [0,1]$ Let: 
	\begin{align}
	p_{c,x,y,u}&\triangleq Q_Y\BRAs{d(x,Y) < d(x,y)} \notag \\ & +u\cdot Q_Y\BRAs{d(x,Y) = d(x,y)} 
	\end{align}
	The pairwise correct decoding probability is the random variable: $p_{c,x,y,U}$ where $U\sim \cU$ is uniform over $[0,1]$. 
\end{definition}
The following proposition summaries the properties we need about the pairwise correct decoding $p_{c,x,y,u}$:
\begin{proposition} \label{proposition:pairwise_correct} \mynewline
	\begin{enumerate}
		\item For any $w\in[0,1]$ and $x$ there exist $y$ and $\tau$ such that: $$w=p_{c,x,y,\tau}.$$
		\item $d(x,y_1) < d(x,y_2) \Rightarrow p_{c,x,y_1, U_1} \leq p_{c,x,y_2, U_2}$. 
		If $Q_Y(y_1)>0$ or $Q_Y(y_2)>0$ then $p_{c,x,y_1, U_1} < p_{c,x,y_2, U_2}$ with probability 1. 
		\item $p_{c,x,Y,U} \sim W$ where $Y \sim Q_Y$ and $U,W$ are uniform over $[0,1]$. 
	\end{enumerate}
\end{proposition}

\ifFullProofs

\begin{proof}
	To prove (1) note that there must exist a $y$ such that:
	\begin{align*}
	Q_Y\BRAs{d(x,Y) < d(x,y)} \leq w &\leq Q_Y\BRAs{d(x,Y) \leq d(x,y)}
	\end{align*}
	If $Q_Y\BRAs{d(x,Y) = d(x,y)}=0$ then we are done with any $\tau$. If $Q_Y\BRAs{d(x,Y) = d(x,y)}\neq 0$ then $$\tau = \frac{w-Q_Y\BRAs{d(x,Y) < d(x,y)}}{Q_Y\BRAs{d(x,Y) = d(x,y)}}\leq 1$$ satisfies the requirement. 
	To prove (2):
	\begin{align*}
	p_{c,x,y_1, U_1} &= Q_Y\BRAs{d(x,Y) < d(x,y_1)}\\&+U_1\cdot Q_Y\BRAs{d(x,Y) = d(x,y_1)} \\
	&\overset{(a)}{\leq} Q_Y\BRAs{d(x,Y) \leq d(x,y_1)} \\
	&\leq Q_Y\BRAs{d(x,Y) < d(x,y_2)} \\
	&\overset{(b)}{\leq} Q_Y\BRAs{d(x,Y) < d(x,y_2)}\\&+U_2\cdot Q_Y\BRAs{d(x,Y) = d(x,y_2)} \\
	&= p_{c,x,y_2, U_2}
	\end{align*}
	If $Q_Y(y_1) > 0$ then $Q_Y\BRAs{d(x,Y) = d(x,y_1)} > 0$ thus we have strict inequality in (a). If $Q_Y(y_2) > 0$ we have strict inequality in (b). To prove (3) let $y_w$ and $\tau_w$ be such that:
	\begin{align*}
	w &= Q_Y\BRAs{d(x,Y) < d(x,y_w)}\\
	&+\tau_w\cdot Q_Y\BRAs{d(x,Y) = d(x,y_w)}
	\end{align*}
	Then:
\begin{align*}
&\PRs{Q_Y}{p_{c,x,Y,U} < w} \\
&= \sum_{y: d(x,y) < d(x,y_w)}Q_Y(y)\PRs{Q_Y}{p_{c,x,y,U} < w} \\
&+ \sum_{y: d(x,y) = d(x,y_w)}Q_Y(y)\PRs{Q_Y}{p_{c,x,y,U} < w} \\
&+ \sum_{y: d(x,y) > d(x,y_w)}Q_Y(y)\PRs{Q_Y}{p_{c,x,y,U} < w} 
\end{align*}
If $d(x,y) \lessgtr d(x,y_w)$ and $Q_Y(y) > 0$ then: $$p_{c,x,y,U} \lessgtr w=p_{c,x,y_w,\tau}$$ according to (2). Thus the first sum gives $Q_Y\BRAs{d(x,Y) < d(x,y_w)}$ and the last sum vanishes. If $d(x,y) = d(x,y_w)$ and $Q_Y(y)>0$ then: 
\begin{align*}
\PRs{Q_Y}{p_{c,x,y,U} < w} &= \PRs{Q_Y}{p_{c,x,y_w,U} < p_{c,x,y_w,\tau_w}}\\&=\PR{U < \tau_w}=\tau_w
\end{align*}
hence the middle sum gives $\tau_w\cdot Q_Y\BRAs{d(x,Y) = d(x,y_w)}$. Combined:
\begin{align*}
\PRs{Q_Y}{p_{c,x,Y,U} < w} &= Q_Y\BRAs{d(x,Y) < d(x,y_w)}\\&+\tau_w\cdot Q_Y\BRAs{d(x,Y) = d(x,y_w)} \\
&=w
\end{align*}
\end{proof}

\else
The proof is omitted due to space limitation and appears in the full paper available online: \cite{elkayam2019oneshot_rate_distortion_full}.	
\fi

\subsection{Random coding performance}
Proposition \ref{proposition:pairwise_correct} suggests that for any fixed $x$ we have a correspondence between the elements of $\cY$ and the sub interval $$[Q_Y\BRAs{d(x,Y)<d(x,y)}, Q_Y\BRAs{d(x,Y)\leq d(x,y)}]\subset [0,1]$$ given by $y \iff p_{c,x,y,U}$. We define the distortion function $d$ on $\cX\times [0,1]$ according to: 
\begin{equation}\label{distortion_extansion_to_unit_interval}
\tilde{d}(x,u)\triangleq d(x,y),\ \  u=p_{c,x,y,\tau}
\end{equation}
This correspondence is well defined almost everywhere with respect to the pair $\BRA{Q_Y,U}$ since if $d(x,y_1)\neq d(x,y_2)$, and $Q_Y(y_1) > 0$ or $Q_Y(y_2) > 0$ the support of $p_{c,x,y_1,U_1}$ and $p_{c,x,y_2,U_2}$ do not overlap with probability 1. Moreover, this mapping is order preserving, \ie, 
\begin{equation}\label{corresponednce_order_preserving}
d(x,y_1)<d(x,y_2) \Rightarrow p_{c,x,y_1,U_1}\leq p_{c,x,y_2,U_2}
\end{equation}

The following result provides an exact formula for the average distortion of random code. 
\begin{theorem}[Exact performance of random coding]\label{theorem:exact_performance_one_shot_rate_distortion}
	The average distortion of random code with $M=e^R+1$ codewords $\BRAs{Y_i}$ drawn from $Q_Y$ is given by:
	\begin{equation}
		\Es{P_X,\BRAs{Y_i}}{\min d(X,Y_i)} = \int_{0}^{1}\tilde{D}(w,Q_Y)G_M'(w) dw 
	\end{equation}
	where: $\tilde{D}(w,Q_Y)=w^{-1}\cdot \Es{P_X\times Q_Y}{d(X,Y)\cdot\Ind{p_{c,X,Y,U}\leq w}}$ and $G_M(w)=-(1-w)^{M-1}((M-1)w+1)$.
\end{theorem}

\begin{corollary}\label{distortion_rate_upper_bound}
	For any $\lambda<R$:  
	\begin{align*}
		&\E{\min d(X,Y_i)} \leq \tilde{D}\BRA{e^{-(R-\lambda)},Q_Y}\\&+\BRA{\tilde{D}\BRA{1,Q_Y}-\tilde{D}\BRA{e^{-(R-\lambda)},Q_Y}}\cdot e^{-e^{\lambda}}(e^{\lambda}+1) \\
		&\leq \tilde{D}\BRA{e^{-(R-\lambda)},Q_Y}+d_{max}\cdot e^{-e^{\lambda}}(e^{\lambda}+1)
	\end{align*}
\end{corollary}
For $z \in [\tilde{D}(0,Q_Y), \tilde{D}(1,Q_Y)]$ let:
$$\tilde{D}^{-1}(z,Q_Y) \triangleq \inf\BRAs{w \in [0,1]: \tilde{D}(w,Q_Y) \geq z}$$ and: 
\begin{equation}\label{rate_distortion_formula}
\tilde{R}(z,Q_Y) \triangleq -\log \tilde{D}^{-1}(z,Q_Y)
\end{equation}
$\tilde{R}(z,Q_Y)$ is the ``rate distortion'' function associated with the prior distribution $Q_Y$. 
\begin{corollary}\label{distortion_rate_distortion_solution}
	For any $Q_Y$, let $d_{req}\in(\tilde{D}(0,Q_Y), \tilde{D}(1,Q_Y))$ denote a desired distortion level.
	There exist a code with distortion level $d_{req}$ and rate $R$ such that:
	\begin{align*}
	R &\leq \MIN{z < d_{req}} \tilde{R}(z,Q_Y) + f^{-1}\BRA{\frac{d_{req}-z}{\tilde{D}(1,Q_Y)-z}} \\
	&\leq \MIN{z < d_{req}} \tilde{R}(z,Q_Y) + g\BRA{\frac{\tilde{D}(1,Q_Y)-z}{d_{req}-z}} 
	\end{align*}
	where $f(t)=e^{-e^t}(e^t+1)$ and $g(x)=\log\log(x)+\log\BRA{\frac 23}+\log\BRA{1 + \sqrt{1+9\BRA{2\log(x)}^{-1}  } }$
\end{corollary}

\begin{proof}[Proof of theorem \ref{theorem:exact_performance_one_shot_rate_distortion}:]
 To calculate the average distortion of a random code with $M$ codewords we proceed as follow:
\begin{align*}
\E{\min d(X,Y_i)} &\overset{(a)}{=} \Es{P_X}{\E{\min d(X,Y_i) | X}} \\
&\overset{(b)}{=} \Es{P_X}{\E{\MIN{i}\BRAs{ \tilde{d}(X,p_{c,X,Y_i, U_i})|X}}} \\
&\overset{(c)}{=} \Es{P_X}{\E{\tilde{d}(X, \MIN{i}\BRAs{ p_{c,X,Y_i, U_i}})}| X} \\
&\overset{(d)}{=} \Es{P_X}{\Es{W_M}{\tilde{d}(X,W_M)|X}} \\
&\overset{(e)}{=} \Es{W_M}{\Es{P_X}{\tilde{d}(X,W_M)|W_M}} \\
&\overset{(f)}{=} \int_{0}^{1}\Es{P_X}{\tilde{d}(X,w)}f_M(w)dw \\
&\overset{(g)}{=} f_M(1)\tilde{D}_1(1)-f_M(0)\tilde{D}_1(0)\\&-\int_{0}^{1}\tilde{D}_1(w)f_M'(w)dw \\
&\overset{(h)}{=} M\cdot(M-1)\int_{0}^{1}\tilde{D}_2(w)w(1-w)^{M-2} dw \\
&\overset{(i)}{=} \int_{0}^{1}\tilde{D}_2(w)G_M'(w) dw
\end{align*}
where $W_M$ is the minimum of $M$ independent uniform random variables, $\tilde{D}_1(w)=\int_{0}^{w}\Es{X}{\tilde{d}(X,z)}dz$ and $\tilde{D}_2(w)=w^{-1}\cdot \tilde{D}_1(w)$. 
\begin{itemize}
	\item (a) is the law of total expectation with respect to $X$.
	\item (b) follow since $d(x,y)=\tilde{d}(x,p_{c,x,y,U})$ according to \eqref{distortion_extansion_to_unit_interval}.
	\item (c) follow since $p_{c,x,y,U}$ preserve the order induce by $\tilde{d}$ according to \eqref{corresponednce_order_preserving}. 
	\item (d) follow since $p_{c,x,Y_i,U_i}$ are all independent and uniform over $[0,1]$.
	\item (e) is again the law of total expectation with respect to $X$ and $W_M$. 
	\item (f) is the expectation according to the p.d.f $f_M$\footnote{The p.d.f of $W_M$ is $f_M(w)=M\cdot(1-u)^{M-1} $}. 
	\ifFullProofs
	(see \eqref{cdf_of_min} in the appendix). 
	\else
	(see \cite{elkayam2019oneshot_rate_distortion_full} for the proof)
	\fi
	\item (g) is integration by parts.
	\item (h) follow since $\tilde{D}_1(0)=f(1)=0$ and: $$f_M'(w)=-M\cdot(M-1)(1-w)^{M-2}$$
	\item (i) follow since $G_M'(w)=M(M-1)w(1-w)^{M-2}$. 
\end{itemize}
Finlay:
\begin{align*}
\tilde{D}_1(w)&=\int_{0}^{w}\Es{P_X}{\tilde{d}(X,z)}dz \\
&\overset{(a)}{=} \Es{P_X,W}{\tilde{d}(X,W)\cdot\Ind{W \leq w}} \\
&\overset{(b)}{=} \Es{P_X\times Q_Y}{\tilde{d}(X,p_{c,X,Y,U})\cdot\Ind{p_{c,X,Y,U} \leq w}} \\
&\overset{(c)}{=} \Es{P_X\times Q_Y}{d(X,Y)\cdot\Ind{p_{c,X,Y,U}\leq w}} \\
&=w\cdot\tilde{D}(w,Q_Y)
\end{align*}
where (a) follow since $W$ is uniform over $[0,1]$. (b)  follow since for each $x$, $p_{c,x,Y,U}$ is uniform over $[0,1]$ and (c) is \eqref{distortion_extansion_to_unit_interval}.
\end{proof}

\begin{proof}[Proof of corollary \ref{distortion_rate_upper_bound}:]
$\tilde{D}(w,Q_Y)$ is an increasing function, Thus: 
\begin{align*}
\E{\min d(X,Y_i)} &= \int_{0}^{1}\tilde{D}(w,Q_Y)G_M'(w) dw\\
&\leq \tilde{D}(u,Q_Y)\int_{0}^{u}G_M'(w) dw \\&+ \tilde{D}(1,Q_Y)\int_{u}^{1}G_M'(w)dw \\
&= \tilde{D}(u,Q_Y) \\&- \BRA{\tilde{D}(1,Q_Y)-\tilde{D}(u,Q_Y)}G_M(u)
\end{align*}
since $G_M(0)=-1$, $G_M(1)=0$ and: 
\begin{align*}
-G_M(u) &= (1-u)^{M-1}((M-1)u+1) \\
&\leq e^{-u\cdot e^R}(u\cdot e^R+1) = e^{-e^{\lambda}}(e^{\lambda}+1)
\end{align*}
where $u\cdot e^R=e^{\lambda}$. The second bound follow since $\tilde{D}(1,Q_Y) \leq d_{max}$.
\end{proof}
\begin{proof}[Proof of corollary \ref{distortion_rate_distortion_solution}:]
	Using $d_{req}=\E{\min d(X,Y_i)}$ and $z=\tilde{D}\BRA{e^{-(R-\lambda)},Q_Y}$ in corollary \ref{distortion_rate_upper_bound} we have:
	$$ d_{req} \leq z + (\tilde{D}(1,Q_Y)-z)\cdot f(\lambda)$$
	Hence:
	$$ \lambda \leq f^{-1}\BRA{\frac{d_{req}-z}{\tilde{D}(1,Q_Y)-z}}$$
	since $f$ is decreasing. The result follow from:
	$$ \tilde{R}(z,Q_Y) = R-\lambda$$
	The second bound follow from the technical lemma 
	\ifFullProofs
	\ref{solutionToFunction}, given in the appendix.
	\else
	given in the full paper \cite{elkayam2019oneshot_rate_distortion_full}.
	\fi

\end{proof}

The standard rate distortion analysis usually employ a ``test'' channel $W_{Y|X}$ that is used for change of measure during the achievability proof and also serves as the encoding function when the converse is proved. The following Proposition suggest such a channel to be used in our achievability theorem. 
\begin{proposition}\label{proposition_def_rd_channel}
	Let $W_{Y|X=x}^w\triangleq w^{-1}\cdot Q_Y\cdot\Ind{p_{c,x,Y,U}\leq w}$, \ie $$ W_{Y|X=x}^w(y) = w^{-1}\cdot Q_Y(y)\cdot\PRs{U}{p_{c,x,y,U}\leq w}$$ Then $W_{Y|X=x}^w$ is a probability distribution over $\cY$ and:
	$$ \tilde{D}(w,Q_Y)=\Es{P_X\times W_{Y|X}^w}{d(X,Y)}$$
\end{proposition}

\ifFullProofs

\begin{proof}
	$W_{Y|X=x}^w$ is a distribution since:
	\begin{align*}
	w &= \PRs{Q_Y,U}{p_{c,x,Y,U}\leq w} \\
	&= \Es{Q_Y, U}{\Ind{p_{c,x,Y,U}\leq w}}
	\end{align*}
	and:
	\begin{align*}
	\tilde{D}(w,Q_Y)&=\Es{P_X\times Q_Y}{d(X,Y)\cdot w^{-1}\cdot\Ind{p_{c,x,Y,U}\leq w}} \\
	&= \Es{P_X\times W_{Y|X}^w}{d(X,Y)}
	\end{align*}
\end{proof}

\else
The proof is omitted due to space limitation and appears in the full paper available online: \cite{elkayam2019oneshot_rate_distortion_full}.	
\fi

\section{Converse bound}
The channel $W_{Y|X}^w$ that was suggested in proposition \ref{proposition_def_rd_channel} with the proper distribution $Q_Y$ and parameter $w=M^{-1}$ can serve as the encoding function as the next proposition show.
\begin{proposition}
	Let $\cC\subset\cY$ be a code with $M$ codewords. Let $\tilde{W}_{Y|X}$ denote the optimal encoding of the $\cX$ to $\cC$, \ie
	$$\tilde{W}_{Y|X} = \argmin_{c\in \cC}\BRA{d(X,c)}$$
	where ties are broken arbitrary. Let $Q_Y^{\cC}$ distribute uniformly over the codewords, \ie\ $Q_Y^{\cC}(y)=M^{-1}$ if $y\in \cC$ and $Q_Y^{\cC}(y)=0$ otherwise. Then:
	$$\tilde{W}_{Y|X=x}=W_{Y|X=x}^{M^{-1}}$$
\end{proposition}

\ifFullProofs
\begin{proof}
	Recall that $$p_{c,x,y,U}=Q_Y^{\cC}\BRA{d(x,Y)<d(x,y)}+U\cdot Q_Y^{\cC}\BRA{d(x,Y)=d(x,y)}.$$

\underline{Case I: $Q_Y^{\cC}\BRA{d(x,Y)<d(x,y)}>0$}: In this case there exist $y' \neq y$ such that $d(x,y')<d(x,y)$. Thus: $\tilde{W}_{Y|X=x}(y|x)=0$ since $x$ cannot encode to $y$ ($y'$ produce lower distortion). From $Q_Y^{\cC}\BRA{d(x,Y)<d(x,y)}\geq M^{-1}$ it follow that $p_{c,x,y,U} > M^{-1}$ with probability 1 and $W_{Y|X=x}^{M^{-1}}(y|x)=0$ as well. 
	
	\underline{Case II: $Q_Y^{\cC}\BRA{d(x,Y)<d(x,y)}=0$}: Let: $$Q_Y^{\cC}\BRA{d(x,Y)=d(x,y)}=k\cdot M^{-1}$$ where $k$ is integer greater than 0. There are $k$ different symbols $y_1=y,\dots, y_k \in \cY$ such that $d(x,y_i)=d(x,y)$ and 
	$\tilde{W}_{Y|X=x}(\cdot|x)$ encode $x$ to one of these symbols randomly. Thus: $$\tilde{W}_{Y|X=x}(y_1|x)=\dots=\tilde{W}_{Y|X=x}(y_k|x)=k^{-1}$$
	in particular, $\tilde{W}_{Y|X=x}(y|x)=k^{-1}$. Since: $$p_{c,x,y,U}=U\cdot Q_Y^{\cC}\BRA{d(x,Y)=d(x,y)}=U\cdot k\cdot M^{-1}$$
	we have:
	\begin{align*}
	\PRs{U}{p_{c,x,y,U} \leq M^{-1}}&=\PRs{U}{U\cdot k\cdot M^{-1} \leq M^{-1}} \\
	&=k^{-1}
	\end{align*}
	thus:
	$$W_{Y|X=x}^{M^{-1}}(y) = M\cdot Q_Y^{\cC}(y)\cdot\PRs{U}{p_{c,x,y,U}\leq M^{-1}}=k^{-1}$$
	as required. 
\end{proof}
\else
The proof is omitted due to space limitation and appears in the full paper available online: \cite{elkayam2019oneshot_rate_distortion_full}.	
\fi

\begin{theorem}
	For any code $\cC$:
	$$ D(P_X, \cC) =\tilde{D}(M^{-1},Q_Y^{\cC})$$
	in particular:
	$$ D(P_X, R) \geq \inf_{Q_Y\in\scP(\cY)} \tilde{D}(e^{-R},Q_Y)$$
\end{theorem}
\begin{proof}
	\begin{align*}
	D(P_X, \cC) 
	&= \Es{P_X\times \tilde{W}_{Y|X}}{d(X,Y)} \\
	&=\Es{P_X\times W_{Y|X}^{M^{-1}}}{d(X,Y)} \\
	&=\tilde{D}(M^{-1},Q_Y^{\cC})
	\end{align*}
\end{proof}

Let:
$$ \hat{D}(z) \triangleq \INF{Q_Y}\tilde{D}(e^{-z},Q_Y) $$
Writing the achievable and converse results in terms of $\hat{D}(z)$, we have: 
\begin{align}
\hat{D}(R) &\leq D(P_X,R) \\&\leq \INF{\lambda < R}\BRAs{\hat{D}(R-\lambda)+d_{max}\cdot e^{-e^{\lambda}}(e^{\lambda}+1) } 
\end{align}
This equation exemplify the tightness of our bounds. 

\section{Variational forms of $\tilde{D}(z,Q_Y)$}
In this section we present two alternative presentations of $\tilde{D}(z,Q_Y)$. We first recall some definition. The definition of $\beta_{\alpha}\BRA{P,Q}$ which represent the optimal performance of a binary hypothesis testing between two $\sigma-$finite measures $P$ and $Q$ over a set $W$:
\begin{equation}\label{binary_hypothsis:beta}
\beta_{\alpha}\BRA{P,Q} = \MIN{\substack{P_{Z|W} :\\ \sum_{w\in W}P(w)P_{Z|W}(1|w) \geq \alpha} } \sum_{w\in W}Q(w)P_{Z|W}(1|w),
\end{equation}
$P_{Z|W}:W \rightarrow \BRAs{0,1}$ is any randomized test between $P$ and $Q$. The minimum is guaranteed to be achieved by the Neyman--Pearson lemma. 
The functional $\beta_{\alpha}\BRA{P,Q}$ has been proved useful for converse results in channel coding and lossy compression, \eg\ \cite[Theorem 26]{polyanskiy2010channel}, \cite[Theorem 8]{DBLP:journals/tit/KostinaV13}. 

The $\infty$-order divergences between $P$ and $Q$ is:
\begin{equation}
D_{\infty}(P||Q)\triangleq \log\inf\BRAs{\lambda : P(x)\leq \lambda Q(x), \forall x}
\end{equation}


\begin{theorem}[Variational forms of $\tilde{D}(z,Q_Y)$]\label{theorem:variational_form} Let $w=e^{-R}$. Then:
\begin{align}
&\tilde{D}(z,Q_Y) \notag\\&= \SUP{Q_X}\mbox{ }\beta_{w}\BRA{Q_X\times Q_Y, P_X\times Q_Y\times d(X,Y)} \label{variational_formula:sup}\\
&= \INF{W_{Y|X}: D_{\infty}(P_X\times W_{Y|X}||P_X\times Q_Y) \leq R}\Es{P_X\times W_{Y|X}}{d(X,Y)} \label{variational_formula:inf}
\end{align}
\end{theorem}
\begin{corollary}
	The convexity of $\tilde{D}(z,Q_Y)$ with respect to $Q_Y$ readily follows from \eqref{variational_formula:sup}. Since $\beta_{w}\BRA{Q_X\times Q_Y, P_X\times Q_Y\times d(X,Y)}$ is convex with respect to $Q_Y$ according to \cite[Theorem 6]{polyanskiy2013saddle} and supremum of convex function is convex as well. 
\end{corollary}

\ifFullProofs

\begin{proof}
	To prove \eqref{variational_formula:sup}, let $\Ind{p_{c,X,Y,U}\leq w}$ denote a (not necessarily optimal) test between $Q_X\times Q_Y$ and $P_X\times Q_Y\times d(X,Y)$. Since $\PRs{Q_Y}{p_{c,x,Y,U}\leq w}=w$ for any $x$, it follow that for any $Q_X$ we have: $\PRs{Q_X\times Q_Y}{p_{c,X,Y,U}\leq w}=w$:
\begin{align*}
&\beta_{w}\BRA{Q_X\times Q_Y, P_X\times Q_Y\times d(X,Y)} \\&\leq \Es{P_X\times Q_Y\times d(X,Y)}{\Ind{p_{c,X,Y,U}\leq w}} \\
&= \Es{P_X\times Q_Y}{d(X,Y)\cdot \Ind{p_{c,X,Y,U}\leq w}}
\end{align*}
Thus:
\begin{align*}
&\sup_{Q_X}\beta_{u}\BRA{Q_X\times Q_Y, P_X\times Q_Y\times d(X,Y)} \\&\leq \Es{P_X\times Q_Y}{d(X,Y)\cdot \Ind{p_{c,X,Y,U}\leq u}}
\end{align*}
To show the reverse inequality, we show that there exist $Q_X^c$ such that:
\begin{align*}
&\beta_{w}\BRA{Q_X^c\times Q_Y, P_X\times Q_Y\times d(X,Y)} \\&= \Es{P_X\times Q_Y}{d(X,Y)\cdot \Ind{p_{c,X,Y,U}\leq w}}
\end{align*}
which will complete the proof. We will construct such a $Q_X^c$ and show that the optimal likelihood test between $Q_X^c\times Q_Y$ and $P_X\times Q_Y\times d(X,Y)$ matches the test $\Ind{p_{c,X,Y,U}\leq w}$. 
The likelihood ratio is: $$L(x,y)=\frac{P_X(x)\times Q_Y(y)\cdot d(x,y)}{Q_X^c(x)\times Q_Y(y)}=\frac{P_X(x)}{Q_X^c(x)}\times d(x,y)$$ and the likelihood ratio test is:
$$P(z|x,y)=\Ind{L(x,y) < \lambda} + \sum_{x'} \tau_{x'} \Ind{L(x,y) = \lambda, x=x'}$$
where $\lambda$ and $\tau_{x'}$ are tuned so that:
$$ \PRs{Q_X^c\times Q_Y}{P(z|X,Y) = 1}=w$$

For any $x$ there exist $y_x$ and $\tau_x$ such that:
$$ w = Q_Y(d(x,Y)< d(x,y_x)) +\tau_x\cdot Q_Y(d(x,Y)= d(x,y_x))$$
Define $\lambda$ and $Q_X^c$:
\begin{align*}
\lambda &= \sum_x P_X(x)\cdot d(x,y_x) \\
Q_X^c(x) &= \lambda^{-1}\cdot P_X(x)\cdot d(x,y_x)
\end{align*}
Note that $Q_X^c$ is probability distribution and:
\begin{align*}
L(x,y_x) = \frac{P_X(x)}{Q_X^c(x)}\cdot d(x,y_x) &= \lambda
\end{align*}
To show that the tests matches we have to prove:
\begin{align*}
\PR{p_{c,x,y,U} \leq w} &= \PR{P(z|x,y) = 1} 
\end{align*}
for any $x$ and $y$. There are 3 cases to consider:
\begin{enumerate}
	\item $d(x,y) < d(x,y_x)$: In this case:
	\begin{align*}
	p_{c,x,y,U} &\leq Q_Y(d(x,Y)\leq d(x,y)) \\&\leq Q_Y(d(x,Y)< d(x,y_x)) \\&\leq w
	\end{align*}
	thus: $\PR{p_{c,x,y,U} \leq w}=1$. On the other hand, since:
	\begin{align*}
	L(x,y)&=\frac{P_X(x)}{Q_X^c(x)}\cdot d(x,y) \\&< \frac{P_X(x)}{Q_X^c(x)}\cdot d(x,y_x) \\&= \lambda
	\end{align*}
	 we also have: $\PR{P(z|x,y) = 1}=1$. 
	\item $d(x,y) > d(x,y_x)$: In this case: 
	\begin{align*}
	p_{c,x,y,U} &\geq Q_Y(d(x,Y)< d(x,y)) \\&\geq Q_Y(d(x,Y)\leq  d(x,y_x)) \\&\geq w
	\end{align*}
	thus: $\PR{p_{c,x,y,U} \leq w}=0$. On the other hand, since 
	\begin{align*}
	L(x,y)&=\frac{P_X(x)}{Q_X^c(x)}\cdot d(x,y) \\&> \frac{P_X(x)}{Q_X^c(x)}\cdot d(x,y_x) \\&= \lambda
	\end{align*}
	
	we also have: $\PR{P(z|x,y) = 1}=0$. 	
	\item $d(x,y) = d(x,y_x)$: In this case 
	\begin{align*}
	\PR{p_{c,x,y,U}\leq w} &=\PR{p_{c,x,y_x,U}\leq w}\\
	&=\tau_x
	\end{align*} 
	On the other hand, since 
	\begin{align*}
	L(x,y)&=\frac{P_X(x)}{Q_X^c(x)}\cdot d(x,y) \\&= \frac{P_X(x)}{Q_X^c(x)}\cdot d(x,y_x) \\&= \lambda
	\end{align*}
	we also have: $\PR{P(z|x,y) = 1}=\tau_x$.
\end{enumerate}
Thus, the test $\Ind{p_{c,x,y,U}\leq w}$ matches the likelihood ratio test and is, in fact, optimal.

	To prove \eqref{variational_formula:inf}, 
	recall the channel: $$W_{Y|X=x}^{e^{-R}}=e^{R}\cdot Q_Y \cdot \Ind{p_{c,x,Y,U}\leq e^{-R}}$$ For any $y$ such that $Q_Y(y)>0$:
	$$\log\frac{e^{R}\cdot Q_Y(y) \cdot \PR{p_{c,x,Y,U}\leq e^{-R}}}{Q_Y(y)}\leq R$$
	Thus: $D_{\infty}(P_X\times W_{Y|X}^{e^{-R}}||P_X\times Q_Y) \leq R$ and:
	\begin{align*}
	&\tilde{D}(e^{-R}, Q_Y)\\&\geq\MIN{W_{Y|X}: D_{\infty}(P_X\times W_{Y|X}||P_X\times Q_Y) \leq R}\Es{P_X\times W_{Y|X}}{d(X,Y)}
	\end{align*}

	On the other hand, if $W_{Y|X}$ satisfy: $$D_{\infty}(P_X\times W_{Y|X}||P_X\times Q_Y) \leq R$$ we have $W_{Y|X}(y|x)P_X(x)\leq e^R\cdot Q_Y(y)P_X(x)$ for each $x$ and $y$. Thus, to get the minimal $\Es{P_X\times W_{Y|X}}{d(X,Y)}$ we will have to assign the maximal probability to the minimal distortion, \ie\ $Q_Y(y)\cdot e^R$. This is exactly what the assignment $Q_Y\cdot\Ind{p_{c,X,Y,U}\leq e^{-R}}\cdot e^R$ does which assign $Q_Y(y)\cdot e^R$ for the smallest distortion values until it exhaust the probability to one. 
\end{proof}
\else
The proof is omitted due to space limitation and appears in the full paper available online: \cite{elkayam2019oneshot_rate_distortion_full}
\fi

Let $W_{Y|X}$ denote any channel and let $Q_Y$ denote the marginal distribution of $W_{Y|X}\times P_X$. Let $i(x;y)=\frac{W_{Y|X}(y|x)}{Q_Y(y)}$ denote the information density. Note that in theorem \ref{theorem:variational_form} we did not require that $Q_Y$ is the marginal distribution. We might have relaxed the requirement $D_{\infty}(P_X\times W_{Y|X}||P_X\times Q_Y) \leq R$ which amounts to $i(x;y)\leq R$ for each $x$ and $y$ when $Q_Y$ is the marginal distribution to the requirement that $\PRs{P_X\times W_{Y|X}}{i(X,Y)\leq R-\delta}$ is close to 1. 
\begin{theorem}\label{relation_Dr_classic}
	Let $W_{Y|X}$ such that: $$\PRs{P_X\times W_{Y|X}}{i(X,Y)\leq R-\delta}=e^{-\lambda}$$ Then:
\begin{align*}
\tilde{D}(e^{-(R-\delta)-\lambda}, Q_Y) &\leq \Es{P_X\times W_{Y|X}}{d(X,Y)}\cdot e^{\lambda} 
\end{align*}
\end{theorem}

\ifFullProofs

\begin{proof}
\begin{align*}
\Es{P_X\times W_{Y|X}}{d(X,Y)} &\geq \Es{P_X\times W_{Y|X}}{d(X,Y)\Ind{i(x,y)\leq R-\delta}} \\ 
&\geq \Es{P_X\times W_{Y|X}'}{d(X,Y)}e^{-\lambda} \\
&\overset{(a)}{\geq} \tilde{D}(e^{-(R-\delta)-\lambda}, Q_Y)\cdot e^{-\lambda}
\end{align*}
where $W_{Y|X}'=e^{\lambda}\cdot W_{Y|X}\Ind{i(x,y)\leq R-\delta}$ is a probability distribution. Note that:
$$\frac{W_{Y|X}'(y|x)}{Q_Y(y)} = \frac{e^{\lambda}\cdot W_{Y|X}(y|x)\Ind{i(x,y)\leq R-\delta}}{Q_Y(y)}\leq e^{R-\delta}\cdot e^{\lambda}$$
Thus $D_{\infty}(P_X\times W_{Y|X}'||P_X\times Q_Y) \leq R-\delta+\lambda$ and (a) follow from \eqref{variational_formula:inf}. 
\end{proof}

\else
The proof is omitted due to space limitation and appears in the full paper available online: \cite{elkayam2019oneshot_rate_distortion_full}.	
\fi

\section{Excess distortion}
The excess distortion is a spacial case of the average distortion that we have analyzed. Let $d_{th}$ denote the target distortion level, replacing $d(x,y)$ with $\Ind{d(x,y)>d_{th}}$. Let: 
\begin{align}
\tilde{D}(R,d_{th},Q_Y) &= \Es{P_X\times Q_Y}{\Ind{d(X,Y)>d_{th}}\Ind{p_{c,X,Y,U}<e^{-R}}}
\end{align} 
Note that $p_{c,x,y,u}$ is also defined with respect to the ``new'' distortion:
\begin{align*}
	p_{c,x,y,u} &= Q_Y\BRA{\Ind{d(x,Y)>d_{th}}<\Ind{d(x,y)>d_{th}}}\\&+u\cdot Q_Y\BRA{\Ind{d(x,Y)>d_{th}}<\Ind{d(x,y)>d_{th}}}
\end{align*} 
equation \eqref{variational_formula:inf} translates in this case to:
\begin{align*}
&\tilde{D}(R,d_{th},Q_Y) \\
&=\inf_{\substack{W_{Y|X}: \\D_{\infty}(P_X\times W_{Y|X}||P_X\times Q_Y)\leq R}}\PRs{P_X\times W_{Y|X}}{d(X,Y)>d_{th}}
\end{align*} 
and \eqref{rate_distortion_formula} is:
\begin{align*}
&\tilde{R}(\delta,z,Q_Y) \\
&=\inf_{\substack{W_{Y|X}: \\\PRs{P_X\times W_{Y|X}}{d(X,Y)>z} \leq \delta}}D_{\infty}(P_X\times W_{Y|X}||P_X\times Q_Y)
\end{align*} 

\section{Relation to known bounds}
The information spectrum approach \cite[Theorem 5.5.1]{koga2002information} provides a general formula for the rate distortion function. The general formula takes any channel $W_{Y|X}$ and with slight abuse of notation, say that the distortion $\Es{P_X\times W_{Y|X}}{d(X,Y)}$ is achievable for a code with rate $R$ such that $\PRs{P_X\times W_{Y|X}}{i(x;y)\leq R}$ approach 1 where $i(x;y)$ is the information density. The achievability is proved by the random coding argument where the distribution used to draw the codewords is the marginal distribution of $P_X\times W_{Y|X}$ on $\cY$. In this paper (see also \cite{matsuta2015non}) we started with any distribution $Q_Y$ and analyzed the random code performance with no channel in mind. For the achievable part, theorem \ref{relation_Dr_classic} provides the link between the functional $\tilde{D}(z,Q_Y)$ and the elements in the information spectrum formula. The converse follow since a code $\cC$ with rate $R$ satisfies $$D_{\infty}\BRA{P_X\times \tilde{W}_{Y|X}||P_X\times Q_Y^{\cC}}=R$$ where $\tilde{W}_{Y|X}$ and $Q_Y^{\cC}$ were defined in the text. 

Much attention has been given to the problem of lossy compression with the excess distortion constraint. The tightest results (to the best of our knowledge) appeared in \cite{matsuta2015non}. The converse bound \cite[Theorem 2]{matsuta2015non} was shown to be tighter than \cite[Theorem 8]{kostina2012fixed}. In \cite{palzer2016converse} the author demonstrated how the bound \cite[Theorem 8]{kostina2012fixed} can be relaxed to all other bounds presented there. 

The approach in \cite{matsuta2015non} for the achievability was to analyze the random coding for a given prior distribution $Q_Y$ and bound the performance from above using a change of measure from $P_X\times Q_Y$ to $P_X\times W_{Y|X}$ where the marginal distribution of $P_X\times W_{Y|X}$ with respect to $\cY$ does not necessarily matches $Q_Y$. Later they optimize for a given channel $W_{Y|X}$ the best prior distribution. Thus, their bounds are given as an optimization over the a set of channels. In \cite[Theorem 2]{matsuta2015non} the author defined: 
$$M(P_{XY})\triangleq \sum_{y\in\cY}\SUP{x\in\cX: P_{XY}(x,y)>0}P_{Y|X}(y|x)$$
and \cite[Lemma 4]{matsuta2015non} reads:
$$\INF{Q_Y\in\scP(\cY)}D_{\infty}(P_{XY}||P_X\times Q_Y) = \log M(P_{XY})$$
Hence:
\begin{align*}
&\inf_{Q_Y\in\scP(\cY)}\tilde{R}(\delta,d_{th},Q_Y) \\
&=\inf_{\substack{W_{Y|X}: \\\PRs{P_X\times W_{Y|X}}{d(X,Y)>d_{th}} \leq \delta}}\log M(P_X\times W_{Y|X})
\end{align*} 
and our converse bound matches theirs. Their achievability bound (Theorem 2) is slightly tighter than the bound in corollary \ref{distortion_rate_distortion_solution}. While their bound reads:
\begin{align*}
R &\leq \MIN{\delta < \epsilon } \tilde{R}(\delta, d_{th}) + \log\log \BRA{\frac{1-\delta}{\epsilon-\delta}} 
\end{align*}
our bound is:
\begin{align*}
R &\leq \MIN{\delta < \epsilon } \tilde{R}(\delta, d_{th}) + g\BRA{\frac{1-\delta}{\epsilon-\delta}} 
\end{align*}
where: $$g(x)=\log\log(x)+\log\BRA{\frac 23}+\log\BRA{1 + \sqrt{1+9\BRA{2\log(x)}^{-1}  } }.$$ The difference smaller than 1 $nat$ for all practical cases. 
Note that, our bound follow from bounding $$\tilde{D}(z,Q_Y) \leq \tilde{D}(u,Q_Y)\Ind{z\leq u}+\tilde{D}(1,Q_Y)\Ind{z> u}$$ and evaluating the exact formula with this bound. 
Using the exact formula directly or bound $\tilde{D}(z,Q_Y)$ using more points would gives tighter bound. Since their achievability also follow from the random coding, obviously the exact formula is tighter.

%
%
%
%
%
\section{Prior optimization}
The prior optimization problem is the main drawback of our approach. While the optimization problem is convex, we do not have a close form solution to the important case of memoryless source (and channel in the channel coding case) and this should be further investigated. In the channel coding case, the prior optimization problem  for memoryless channel  is ``solved'' by \cite[Theorem 10]{DBLP:journals/tit/VerduH94} which shows that for memoryless channel, we can resort to memoryless priors. For our functional $\tilde{D}(z,Q_Y)$ (and the one used in the channel coding problem) we do not have such a theorem. Moreover, simulation results show that this is not even true, and sometimes, prior with memory are better than the memoryless priors. 

\section{Concluding remarks}
In this paper we presented a novel analysis for the lossy compression problem. We have analyzed the general case of average distortion constraints. We presented tight achievable and converse bounds. Both bounds are given in terms of the functional $\tilde{D}(z,Q_Y)$ which has resemblance to the meta-converse in channel coding. The problem of finding distribution minimizing $\tilde{D}(z,Q_Y)$ is still open in general although the $\tilde{D}(z,Q_Y)$ is convex with respect to $Q_Y$. 

\ifFullProofs
\appendix

\section{Technical lemma's}
\begin{lemma}\label{solutionToFunction}
	For any $x \in (0,1)$ let $\lambda$ be the solution to $x = e^{-e^{\lambda}}(e^{\lambda}+1)$, then: 
	\begin{align}
	\lambda-\log(-\log(x))&\leq \log(\frac 23)+\log\BRA{1 + \sqrt{1-\frac 9{2\log(x)}} } \\
	\lambda-\log(-\log(x))&\geq -\log(2)+\log\BRA{1 + \sqrt{1-\frac 8{\log(x)}}}
	\end{align}
\end{lemma}
\begin{proof}
Let $t=e^{\lambda}$ and $z=-log(x)$. Then:
\begin{align*}
z &= t-\log(1+t)
\end{align*}
Using: $\log(1+t) \leq \frac{t(6+t)}{2(3+2t)}$ (\cite[Eq. 22]{topsok2006some}) we have:
\begin{align*}
z &\geq t-\frac{t(6+t)}{2(3+2t)}\\
&= \frac{t^2}{2+\frac 43t}
\end{align*}
For $t,z \geq 0$ the only solution to: $z=\frac{t^2}{2+\frac43 t}$ is: $t=\frac{2}3z\BRA{1 + \sqrt{1+\frac 9{2z}}}$. 
Thus:
\begin{align*}
z &\geq \frac{3t^2}{2(3+2t)}; t,z\geq 0 \iff \\
t &\leq \frac{2}3z\BRA{1 + \sqrt{1+\frac 9{2z}}} 
\end{align*}
Hence:
\begin{align*}
\lambda &= \log(t) \\
&\leq \log\BRA{\frac{2}3z\BRA{1 + \sqrt{1+\frac 9{2z}}} } \\
&= \log(-\log(x)) + \log(\frac 23)+\log\BRA{1 + \sqrt{1-\frac 9{2\log(x)}} } \\
\end{align*}
The lower bound follow along the same line using $\log(1+t)\geq \frac{2t}{2+t}$. 
\end{proof}

\begin{proposition}
Let $U_i$ denote $M$ independent uniform random variables and let $W_M=\MIN{i}\BRAs{U_i}$. The c.d.f of $W_M$ is:
\begin{equation}\label{cdf_of_min}
f_M(w)=M\cdot(1-u)^{M-1}
\end{equation}
\end{proposition}
\begin{proof}
\begin{align*}
\PR{W_M \leq w} &= 1-\PR{\MIN{i}\BRAs{U_i} > w} \\
&= 1-\prod_{i=1}^{M}\PR{U_i > w} \\
&= 1-(1-w)^M
\end{align*}
Thus, the p.d.f. of $W_M$ is $f_M(w)=M\cdot(1-u)^{M-1}$.
\end{proof}
\else
\fi

\bibliographystyle{IEEEtran}
\bibliography{bib}
\end{document}